\documentclass[11pt]{article}
\usepackage{a4wide}
\usepackage[utf8]{inputenc}

\usepackage{amsthm}
\usepackage{amssymb}
\usepackage{amsmath}

\usepackage{caption,subcaption}
\usepackage{hyperref}
\usepackage[nameinlink, capitalise, noabbrev]{cleveref}
\newtheorem{theorem}{Theorem}[section]
\newtheorem{lemma}[theorem]{Lemma}

\newtheorem{definition}[theorem]{Definition}
\newtheorem{claim}{Claim}

\usepackage{mathtools}

\usepackage{tikz}
\usetikzlibrary{calc}
\usetikzlibrary{hobby}
\pgfdeclarelayer{background}
\pgfdeclarelayer{foreground}
\pgfsetlayers{background,main,foreground}
\usepackage{graphicx}
\graphicspath{ {./images/} }

\usepackage{caption}
\usepackage{subcaption}
\usepackage[noadjust]{cite}
\usepackage{enumitem}
\usepackage{xspace}
\usepackage{xcolor}
\usepackage{xparse}
\usepackage{xfrac}
\usepackage{todonotes}
\usepackage{textpos}
\usepackage{tcolorbox}
\usepackage{dsfont}

\hypersetup{
	colorlinks=true,
	linkcolor=myBlue,
	citecolor=myBlue,
	urlcolor=myBlue,
	bookmarksopen=true,
	bookmarksnumbered,
	bookmarksopenlevel=2,
	bookmarksdepth=3
}

\usepackage{macros}

\begin{document}

\title{Simpler and faster algorithms for detours in planar digraphs
\thanks{This work is a part of project BOBR (KM, MP, MS) that has received funding from the European Research Council (ERC) under the European Union’s Horizon 2020 research and innovation programme (grant agreement No. 948057).
M.~Hatzel was supported by the Federal Ministry of Education and Research (BMBF) and by a fellowship within the IFI programme of the German Academic Exchange Service (DAAD).}
}

\author{
Meike Hatzel\thanks{National Institute of Informatics, Tokyo, Japan (\href{mailto:meikehatzel@nii.ac.jp}{\texttt{meikehatzel@nii.ac.jp}})}
\and Konrad Majewski\thanks{Institute of Informatics, University of Warsaw, Poland (\href{mailto:k.majewski@mimuw.edu.pl}{\texttt{k.majewski@mimuw.edu.pl}})}
\and Michał Pilipczuk\thanks{Institute of Informatics, University of Warsaw, Poland (\href{mailto:michal.pilipczuk@mimuw.edu.pl}{\texttt{michal.pilipczuk@mimuw.edu.pl}})}
\and Marek Sokołowski\thanks{Institute of Informatics, University of Warsaw, Poland (\href{mailto:marek.sokolowski@mimuw.edu.pl}{\texttt{marek.sokolowski@mimuw.edu.pl}})}
}
\date{}

\maketitle
\thispagestyle{empty}

\begin{abstract}
In the \ddetour problem one is given a digraph $G$ and a pair of vertices $s$ and~$t$, and the task is to decide whether there is a directed simple path from $s$ to $t$ in $G$ whose length is larger than $\dist{G}{s}{t}$.
The more general parameterized variant, \dldetour, asks for a simple $s$-to-$t$ path of length at least $\dist{G}{s}{t}+k$, for a given parameter $k$.
Surprisingly, it is still unknown whether \ddetour is polynomial-time solvable on general digraphs.
However, for planar digraphs, Wu and Wang~[Networks, '15] proposed an $\Oh(n^3)$-time algorithm for \ddetour, while Fomin et al.~[STACS 2022] gave a $2^{\Oh(k)}\cdot n^{\Oh(1)}$-time fpt algorithm for \dldetour.
The algorithm of Wu and Wang relies on a nontrivial analysis of how short detours may look like in a plane embedding, while the algorithm of Fomin et al.~is based on a reduction to the \dpaths{3} problem on planar digraphs.
This latter problem is solvable in polynomial time using the algebraic machinery of Schrijver~[SIAM~J.~Comp.,~'94], but the degree of the obtained polynomial factor is huge.

In this paper we propose two simple algorithms: we show how to solve, in planar digraphs, \ddetour in time $\Oh(n^2)$ and \dldetour in time $2^{\Oh(k)}\cdot n^4 \log n$.
In both cases, the idea is to reduce to the \dpaths{2} problem in a planar digraph, and to observe that the obtained instances of this problem have a certain topological structure that makes them amenable to a direct greedy strategy.
\end{abstract}

 \begin{textblock}{20}(-1.9, 3.1)
   \includegraphics[width=40px]{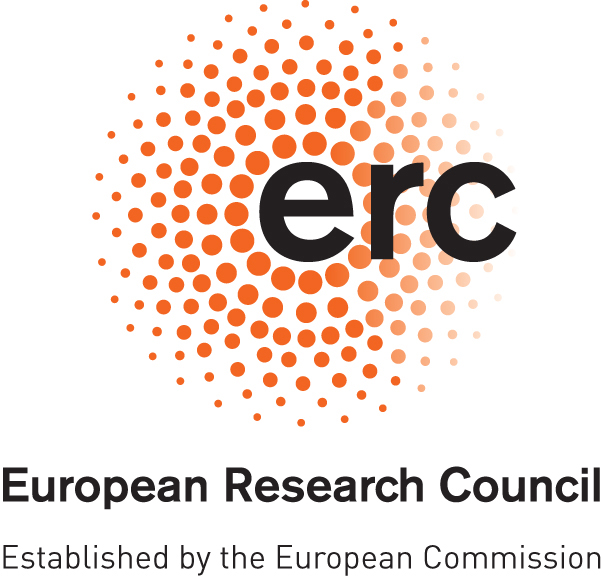}%
 \end{textblock}
 \begin{textblock}{20}(-2.15, 3.5)
   \includegraphics[width=60px]{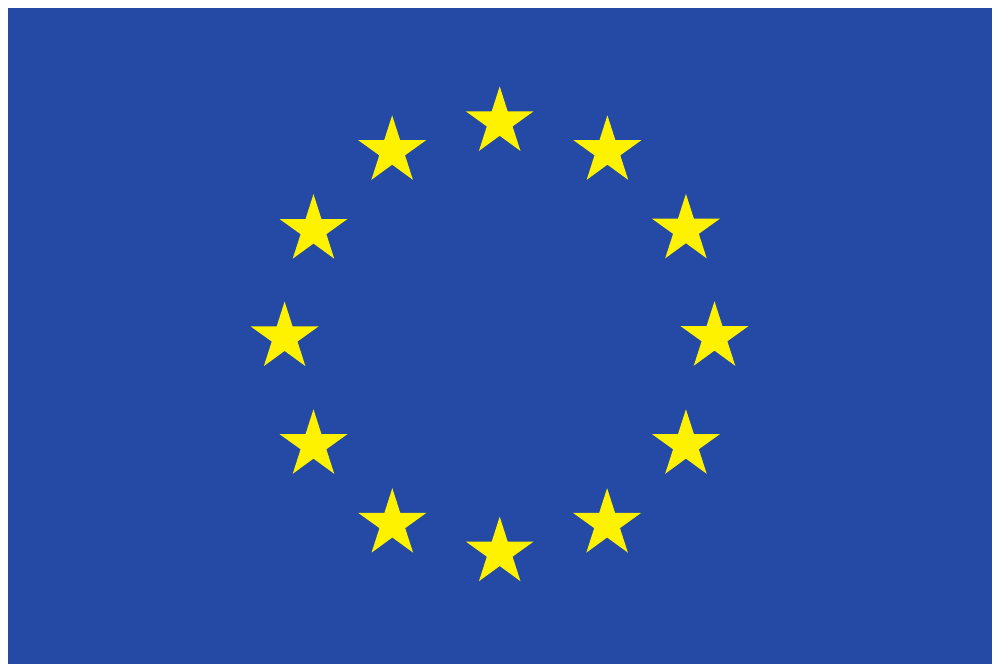}%
 \end{textblock}
 \begin{textblock}{20}(-1.8, 4.8)
	\includegraphics[width=40px]{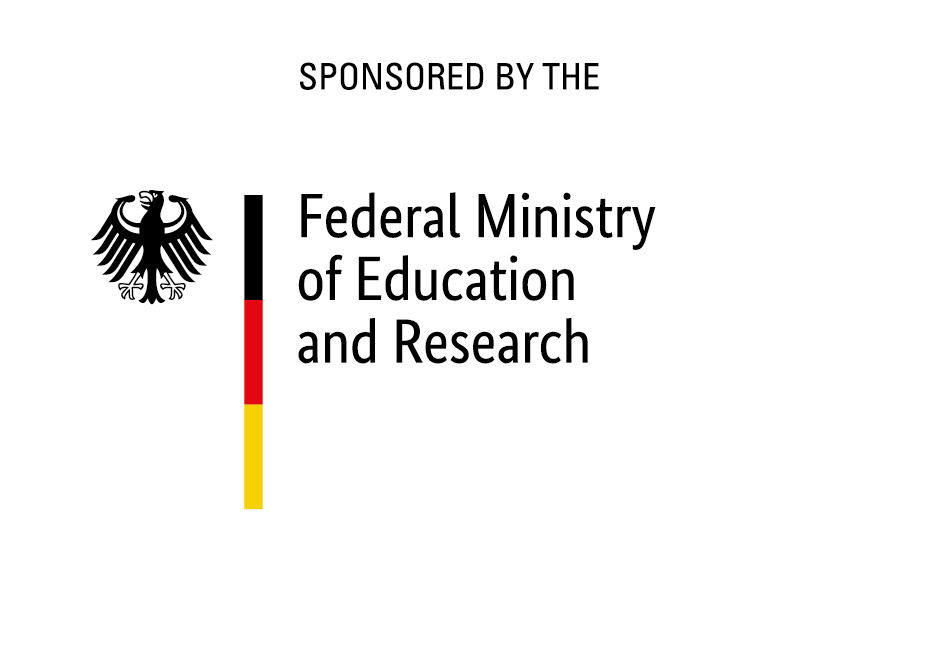}%
\end{textblock}

\newpage
\pagenumbering{arabic}

\section{Introduction}\label{sec:intro}

The complexity status of \ddetour is arguably one of the most tantalizing open questions within the area of graph algorithms. The problem asks to decide, for a given digraph $G$ and two terminals $s$ and $t$, whether there is a simple path in $G$ from $s$ to $t$ that is not the shortest --- has length strictly larger than $\dist{G}{s}{t}$. It is still unknown whether this problem is polynomial-time solvable on general digraphs. See the work of Fomin et al.~\cite{DetoursDirected} for a discussion of relevant literature.

Given this state of the affairs, it is interesting to study \ddetour on restricted classes of digraphs, in hope for finding more positive results or useful insight. In this vein, a particularly well-motivated idea is to consider the class of planar digraphs. The reason for this is that in planar digraphs, the \dpaths{k} problem --- decide the existence of disjoint directed paths linking given $k$ pairs of terminals --- is polynomial-time solvable for every fixed~$k$~\cite{Schrijver94}, and even fixed-parameter tractable when parameterized by $k$~\cite{CyganMPP13}. This is not the case in general digraphs, where the problem is $\NP$-hard already for $k=2$~\cite{FortuneHW80}. This can be used for \ddetour. Namely, Fomin et al.~\cite{DetoursDirected} showed that the more general parameterized version of the problem --- \dldetour, where we look for a simple $s$-to-$t$ path of length at least $\dist{G}{s}{t}+k$, for a given parameter $k$ --- can be reduced (with a $2^{\Oh(k)}\cdot n^{\Oh(1)}$ multiplicative overhead in the complexity) to \dpaths{3}, which can be solved in polynomial time in planar digraph using the algorithm of Schrijver~\cite{Schrijver94}. This of course applies also to the basic \ddetour problem by setting $k=1$, but even earlier Wu and Wang~\cite{WuW15} gave a direct $\Oh(n^3)$-time algorithm for this case.

While the reduction of Fomin et al.~\cite{DetoursDirected} is actually quite simple, the algorithm of Schrijver~\cite{Schrijver94} for \dpaths{3} in planar digraphs is not, as it relies on an involved algebraic framework. In particular, the degree of the polynomial bounding the running time is at least a two-digit number. On the other hand, the cubic algorithm of Wu and Wang~\cite{WuW15} is also quite complicated and relies on an analysis of different planar configurations that may occur.

In this work we propose two simple algorithms: one for \ddetour and one for \dldetour, both in planar digraphs. These are summarized below.

\begin{theorem}\label{thm:main-ddetour}
 \ddetour in planar digraphs can be solved in time $\Oh(n^2)$.
\end{theorem}

\begin{theorem}\label{thm:main-dldetour}
 \dldetour in planar digraphs can be solved in time $2^{\Oh(k)}\cdot n^4$ by a Monte Carlo algorithm, or deterministically in time $2^{\Oh(k)}\cdot n^4 \log n$.
\end{theorem}

The main idea in the proof of Theorem~\ref{thm:main-ddetour} is to perform a reduction to the \dpaths{2} problem, roughly similarly as in the work of Fomin et al.~\cite{DetoursDirected}. However, we observe that if this reduction is performed carefully, then one essentially obtains an instance of \dpaths{2} where three out of four terminals lie on one face, say the outer face. Such instances can be solved very easily: one of the paths --- the one with both terminals on the outer face --- can be chosen greedily so that it leaves the maximum possible space for the other path. Then the other path can be found by a simple reachability check.

For Theorem~\ref{thm:main-dldetour}, as in Fomin et al.~\cite{DetoursDirected}, we use the result of Bez\'akov\'a et al.~\cite{BezakovaCDF19} that the \exdldetour{} problem --- finding a shortest $s$-to-$t$ path of length {\em{exactly}} $\dist{G}{s}{t}+k$ --- can be solved in time $2^{\Oh(k)}\cdot n^2$,
even on general digraphs. This allows us to assume, when solving \dldetour, that there are no $s$-to-$t$ paths of length between $\dist{G}{s}{t}+k$ and $\dist{G}{s}{t}+3k$. Having this assumption, the proof of Theorem~\ref{thm:main-dldetour} proceeds by expanding the basic idea behind Theorem~\ref{thm:main-ddetour} with color coding.

We believe that compared to~\cite{DetoursDirected,WuW15}, our algorithms provide a simpler explanation for the tractability of \dlldetour in planar digraphs. While we do not expect that the gained insight will be directly applicable to the case of general digraphs, we hope that it might be a better starting point for generalizations to less restrictive topological graph classes, for instance to digraphs of bounded genus or digraphs whose underlying undirected graphs exclude a fixed minor.

\section{Preliminaries}\label{sec:prelim}

\paragraph*{Graphs.}
In this paper we consider planar directed graphs $G=(V(G),E(G))$.
For the purposes of our problem, we can assume that the input graphs are simple, that is, they do not have self-loops or multiple arcs connecting two vertices in the same direction.
Moreover, we assume that $G$ is weakly connected, that is, the underlying undirected graph is connected.
For convenience, we set $n = |V(G)|$ and $m = |E(G)|$.
Since $G$ is planar and simple, we have that $m \in \Oh(n)$.

Given an~arc $e = uv \in E(G)$, we say that $u$ is the \emph{tail} of $e$ and $v$ is the \emph{head} of $e$.
Here, we consider $e$ incident to both $u$ and $v$.
A~sequence $v_1 v_2 \dots v_k$ of vertices is called a~\emph{walk} in $G$ if for each $i \in \{1, \dots, k-1\}$, $v_iv_{i+1}$ is an~arc in $G$.
The vertex $v_1$ is called the \emph{origin} of the walk, while $v_k$ is its \emph{destination}.
If the vertices of a~walk are pairwise different, the walk is a~\emph{path}.
Given two walks $W_1$, $W_2$, if the destination of $W_1$ coincides with the origin of $W_2$, then we define $W_1 \circ W_2$ as the concatenation of both walks.
The \emph{length} of a~path $P$, denoted $\length{P}$, is the number of arcs it contains.
Given a path $P$ and two of its vertices $x$ and $y$, we denote by $P[x \to y]$ the subpath of $P$ which starts at $x$, goes along the arcs of $P$, and ends in $y$.

Given a~pair of vertices $u, v \in V(G)$, the \emph{distance} from $u$ to $v$ in $G$ is the length of the shortest $u$-to-$v$-path in $G$ and is denoted by $\dist{G}{u}{v}$.
If there is no directed path from $u$ to $v$ in $G$, we put $\dist{G}{u}{v} = +\infty$.
If the graph $G$ is known from the context, we may omit $G$ from the notation and simply write $\dist{}{u}{v}$.

\paragraph*{Plane embeddings.}
Given a~planar graph $G$, its embedding into the plane is a~mapping of the vertices of $G$ to pairwise distinct points in the plane and the arcs of $G$ to plane curves, so that:
\begin{itemize}
  \item each arc $uv \in E(G)$ is mapped to a~plane curve whose endpoints coincide with the images of $u$ and $v$, and such that no other vertex of $G$ is mapped to a point on the curve; and
  \item the images of the edges of $G$ are pairwise internally disjoint.
\end{itemize}
A~\emph{plane graph} is a~planar graph together with a~fixed embedding of the graph into the plane.
Such an~embedding splits the plane in a~number of regions, called \emph{faces}.
One of the faces is unbounded and called the \emph{outer face}.
Given a~face $F$, its boundary $\bnd F$ can be described as a~cyclic sequence of arcs bounding $F$ (assuming $G$ is connected).
If $\bnd F$ is isomorphic to a~simple polygon (ignoring the orientations of the arcs), we say that $F$ is \emph{simple}.

For algorithmic purposes, we represent planar embeddings combinatorially: each vertex $v$ of the graph stores the anti-clockwise ordering of all arcs of $G$ incident to $v$.
Given a~planar graph $G$, its combinatorial embedding can be computed in time linear with respect to the size of the graph~\cite{DBLP:journals/jacm/HopcroftT74}.
Note that the combinatorial embedding uniquely determines the faces of the planar graph, however, it does not designate the outer face of the embedding. Hence, in our algorithms, we may elect any face to be the outer face.

\paragraph*{Comparing paths in plane graphs.}
Let $G$ be a~plane graph and $F$ be its outer face.
Assume that $F$ is simple.
We say that a~path $P$ in $G$ is \emph{$v$-grounded (with respect to $F$)} if its origin $v$ is a~vertex of $F$.
Moreover, we say that $P$ is \emph{$(v,w)$-grounded (with respect to $F$)} if both its origin $v$ and its destination $w$ are vertices of $F$.
We now present two ways to compare (doubly) grounded paths in~$G$.
We note that the following definitions are standard.

\begin{definition}
  Assume that $P_1$ and $P_2$ are two different $v$-grounded paths with respect to a face $F$ in a plane graph~$G$.
  We say that $P_1$ is \emph{lexicographically left of} $P_2$ (denoted $P_1 \leftlex P_2$) if one of the following conditions holds:
  \begin{itemize}
    \item $P_1$ is a~prefix of $P_2$; or
    \item consider a~plane graph $G_v$ created by taking $G$, together with its planar embedding, and adding a~fresh vertex $v'$ and a~fresh arc $v'v$, mapped to the plane so that the image of $v'$ is placed outside of $F$ (i.e., so that $v'$ is not enclosed by the boundary of $F$).
    	Let $P'_1$ and $P'_2$ be paths in $G_v$, constructed by prepending $v'$ to $P_1$ and $P_2$, respectively.
    	Let $q$ be the last vertex of the longest common prefix of $P'_1$ and $P'_2$, $p$ be the vertex preceding $q$ on the common prefix, and $r_1 \neq r_2$ be the vertices succeeding $q$ on $P'_1$ and $P'_2$, respectively.
    	Then, the arcs $pq$, $qr_1$ and $qr_2$ are embedded clockwise in this order around $q$.
  \end{itemize}
\end{definition}

\begin{figure}
\centering
  \begin{subfigure}[b]{0.49\textwidth}
    \centering
    \begin{tikzpicture}
		\def\height{3.25}
		\def\width{2.25}
    	\node[
    	draw, thick, myGrey,
    	shape=rectangle with rounded corners,
    	minimum height=6.5cm,
    	minimum width = 4.5cm,
    	rectangle with rounded corners north west=30pt,
    	rectangle with rounded corners south west=30pt,
    	rectangle with rounded corners north east=30pt,
    	rectangle with rounded corners south east=30pt,
    	label=45:{\textcolor{myGrey!80!black}{$F$}}] (center) at (0,0) {};
    	
    	\node[vertex,label=270:{$v$}] (v) at ($(center)+(-0.4,-\height)$) {};
    	\node[vertex,label=90:{$w$}] (w) at ($(center)+(0.3,\height)$) {};
    	
    	\definecolor{p1}{rgb}{1.0, 0.49, 0.0}
    	\definecolor{p2}{rgb}{0.55, 0.71, 0.0}
    	\definecolor{p3}{rgb}{0.0, 0.72, 0.92}
    	
    	\node (p-1-1) at ($(center)+(-0.5*\width,-0.6*\height)$) {};
    	\node (p-1-2) at ($(center)+(-0.3*\width,-0.25*\height)$) {};
    	\node (p-1-3) at ($(center)+(-0.75*\width,0.3*\height)$) {};
    	\node (p-1-4) at ($(center)+(-0.3*\width,0.32*\height)$) {};
    	\node (p-1-5) at ($(center)+(-0.32*\width,0.65*\height)$) {};
    	
    	\draw[directededge,p1] (v) to (p-1-1.center);
    	\foreach[evaluate={\j=int(\i+1);}] \i in {1,...,3}
    	{
    		\draw[directededge,p1] (p-1-\i.center) to (p-1-\j.center);
    	}
    	\draw[directededge,p1] (p-1-5.center) to (w);
    	\twocolourededge{p1}{p3}{p-1-4.center}{p-1-5.center}
    	
    	\node (p-2-1) at ($(center)+(-0.15*\width,-0.72*\height)$) {};
    	\node (p-2-2) at ($(center)+(0.17*\width,-0.3*\height)$) {};
    	\node (p-2-3) at ($(center)+(0.02*\width,0.05*\height)$) {};
    	\node (p-2-4) at ($(center)+(0.26*\width,0.33*\height)$) {};
    	\node (p-2-5) at ($(center)+(0.1*\width,0.64*\height)$) {};
    	\node (p-2-6) at ($(center)+(0.4*\width,0.75*\height)$) {};
    	
    	\draw[directededge,p2] (v) to (p-2-1.center);
    	\foreach[evaluate={\j=int(\i+1);}] \i in {2,...,5}
    	{
    		\draw[directededge,p2] (p-2-\i.center) to (p-2-\j.center);
    	}
    	\twocolourededge{p3}{p2}{p-2-6.center}{w};
    	\twocolourededge{p2}{p3}{v}{p-2-1.center}
    	\twocolourededge{p2}{p3}{p-2-1.center}{p-2-2.center}
    	
    	\node (p-3-3) at ($(center)+(0.6*\width,-0.25*\height)$) {};
    	\node (p-3-4) at ($(center)+(0.25*\width,0.01*\height)$) {};
    	
    	\draw[directededge,p3] (p-2-2.center) to (p-3-3.center);
    	\draw[directededge,p3] (p-3-3.center) to (p-3-4.center); 
    	\draw[directededge,p3] (p-3-4.center) to (p-2-4.center); 
    	\draw[directededge,p3] (p-2-4.center) to (p-1-4.center);
    	\draw[directededge,p3] (p-1-5.center) to (p-2-6.center); 
    	
    	\node (p-1-label) at ($(p-1-2)!0.5!(p-1-3)+(-0.3,-0.2)$) {\textcolor{p1!90!red}{$P_1$}};
    	\node (p-2-label) at ($(p-2-3)!0.3!(p-2-4)+(-0.4,0.1)$) {\textcolor{p2!90!black}{$P_2$}};
    	\node (p-3-label) at ($(p-3-3)!0.3!(p-3-4)+(0.3,0.2)$) {\textcolor{p3!90!black}{$P_3$}};

    \end{tikzpicture}
    \caption{
    \label{fig:left_most_def}}
  \end{subfigure}%
  \hfill%
  \begin{subfigure}[b]{0.49\textwidth}
    \centering
    \begin{tikzpicture}
    	\def\height{3.25}
    	\def\width{2.25}
    	\node[
    	draw, thick, myGrey,
    	shape=rectangle with rounded corners,
    	minimum height=6.5cm,
    	minimum width = 4.5cm,
    	rectangle with rounded corners north west=30pt,
    	rectangle with rounded corners south west=30pt,
    	rectangle with rounded corners north east=30pt,
    	rectangle with rounded corners south east=30pt,
    	label=45:{\textcolor{myGrey!80!black}{$F$}}] (center) at (0,0) {};
    	
    	\node[vertex,label=270:{$v$}] (v) at ($(center)+(-0.4,-\height)$) {};
    	\node[vertex,label=90:{$w$}] (w) at ($(center)+(0.3,\height)$) {};
    	
    	\definecolor{p1}{rgb}{1.0, 0.49, 0.0}
    	\definecolor{p2}{rgb}{0.55, 0.71, 0.0}
    	\definecolor{p3}{rgb}{0.0, 0.72, 0.92}
    	
    	\node (p-1) at ($(center)+(0.1*\width,-0.6*\height)$) {};
    	\node (p-2) at ($(center)+(-0.6*\width,-0.59*\height)$) {};
    	\node (p-3) at ($(center)+(-0.75*\width,-0.4*\height)$) {};
    	\node (p-4) at ($(center)+(-1*\width,-0.3*\height)$) {};
    	\node (p-5) at ($(center)+(-1*\width,-0.1*\height)$) {};
    	\node (p-6) at ($(center)+(-0.27*\width,0.1*\height)$) {};
    	\node (p-7) at ($(center)+(-0.1*\width,0.4*\height)$) {};
    	\node (p-8) at ($(center)+(0.5*\width,0.3*\height)$) {};
    	\node (p-9) at ($(center)+(1*\width,0.65*\height)$) {};
    	\node (p-10) at ($(center)+(-0.2*\width,0.73*\height)$) {};
    	
    	\draw[directededge] (v) to (p-1.center);
    	\foreach[evaluate={\j=int(\i+1);}] \i in {1,...,9}
    	{
    		\draw[directededge] (p-\i.center) to (p-\j.center);
    	}
    	\draw[directededge] (p-10.center) to (w);

    	\node (p-label) at ($(p-6)!0.5!(p-7)+(0.3,0)$) {$P$};
    	
    	\begin{pgfonlayer}{background}
    		\begin{scope}
    			    \path [clip] (v.center) -- (p-1.center) -- (p-2.center) -- (p-3.center) -- (p-4.center) -- (p-5.center) -- (p-6.center) -- (p-7.center) -- (p-8.center) -- (p-9.center) -- (p-10.center) -- (w.center) -- ($(0,1.1*\height)$) -- ($(-1.1*\width,1.1*\height)$) -- ($(-1.1*\width,-1.1*\height)$) -- ($(0,-1.1*\height)$) -- cycle;
    				\node[
		    		fill=myLightBlue!80!white,
		    		shape=rectangle with rounded corners,
		    		minimum height=6.5cm,
		    		minimum width = 4.5cm,
		    		rectangle with rounded corners north west=30pt,
		    		rectangle with rounded corners south west=30pt,
		    		rectangle with rounded corners north east=30pt,
		    		rectangle with rounded corners south east=30pt] (filling) at (0,0) {};
    		\end{scope}
    	\end{pgfonlayer}
    \end{tikzpicture}
    \caption{
    \label{fig:left_area}}
  \end{subfigure}
  \hfill
  \caption{\subref{fig:left_most_def} Three $(v,w)$-grounded paths. We have $P_1 \leftlex P_2 \leftlex P_3$, $P_1 \lefttop P_2$, $P_1 \lefttop P_3$, but $P_2\not\lefttop P_3$.
  	\subref{fig:left_area} A $(v,w)$-grounded path $P$ and its left area $\leftarea{P}$ (filled in \textcolor{myLightBlue}{blue}).}
  \label{fig:left-defs}
\end{figure}
Intuitively, $P_1 \leftlex P_2$ if at the end of the common prefix of $P_1$ and $P_2$, the path $P_1$ terminates or branches off left of $P_2$ (Figure~\ref{fig:left_most_def}).
Observe that \leftlex induces a~linear order on all paths grounded at a~vertex $v \in V(F)$.
Thus, given any vertex $w \in V(G)$ reachable from $v$ in $G$, we define $L_{vw}$, the \emph{lexicographically leftmost} path from $v$ to $w$, as the unique minimal path from $v$ to $w$ with respect to \leftlex.
Given a~combinatorial embedding of $G$, we can compute the lexicographically leftmost path from $v$ to $w$ in time $\Oh(n)$: it suffices to run a~depth-first search of $G$ in which each vertex considers all its neighbors in the left-to-right order.

We follow with a~more restrictive way of comparing paths in $G$.
Given a~$(v,w)$-grounded path $P$, we define the \emph{area left of $P$}, denoted $\leftarea{P}$, as the region of the plane whose anti-clockwise boundary is the closed walk defined as the concatenation of $P$ and the anti-clockwise segment of the boundary of $F$ from $w$ to $v$ (Figure~\ref{fig:left_area}).
Note that the interior of $\leftarea{P}$ may be disconnected if $P$ internally intersects the boundary of $F$.

\begin{definition}
  Given two $(v,w)$-grounded paths $P_1$ and $P_2$, we say that $P_1$ is \emph{topologically left of} $P_2$ (denoted $P_1 \lefttop P_2$) if $\leftarea{P_1} \subsetneq \leftarea{P_2}$.
\end{definition}

Note that if $P_1 \lefttop P_2$, then $P_1 \leftlex P_2$.
Thus, $\lefttop$ is a~partial order on the set of all $(v,w)$-grounded paths, and $\leftlex$ restricted to those paths is a linear extension of $\lefttop$.
It is straightforward to observe that the lexicographically leftmost path from $v$ to $w$ --- the minimum $(v,w)$-grounded path in $\leftlex$ --- is also the unique minimum $(v,w)$-grounded path in $\lefttop$.
So if both $v$ and $w$ lie on $F$, we call $L_{vw}$ simply the \emph{leftmost path} from $v$ to~$w$.

Analogously, we define the \emph{(lexicographically) rightmost path} from $v$ to $w$, denoted $R_{vw}$, and the \emph{area right of} a~$(v,w)$-grounded path $P$, denoted $\rightarea{P}$.
Note that each $(v,w)$-grounded path $P$ splits the graph into two parts: the area $\leftarea{P}$ left of $P$ and the area $\rightarea{P}$ right of $P$, with both areas intersecting only at~$P$.

\paragraph*{Simplifying the outer face.}
Both the lexicographical and the topological comparisons of grounded paths require the outer face $F$ of a~plane graph $G$ to be simple.
However, this precondition might not be satisfied in the setting of our problem.
To alleviate this issue, we show how to simplify the outer face by adding two directed arcs to the plane graph.

Let $u$ and $v$ be two vertices of $F$.
The \emph{$(u,v)$-simplification} of the outer face in $G$, denoted $\simplify{G}{u}{v}$, is the digraph obtained from $G$ by adding two arcs $f_1$, $f_2$, each from $u$ to $v$.
Both arcs are embedded in the outer face $F$ of $G$, so that the outer face of $\simplify{G}{u}{v}$ is bounded by $f_1$ and $f_2$ (see Figure~\ref{fig:simplification-example}).
Thus, the outer face of $\simplify{G}{u}{v}$ is simple and contains two vertices, $u$ and $v$.
We remark that $\simplify{G}{u}{v}$ is not necessarily a~simple digraph; however, this does not pose a~problem since after the simplification, the number of edges of the graph is still linearly bounded in the number of vertices of the graph.

\begin{figure}
  \centering
  \begin{subfigure}[h]{0.49\textwidth}
    \centering
    \begin{tikzpicture}{scale=0.8}
    	\def\udist{2}
    	\def\vdist{0.8}
    	
    	\node[vertex] (u) at (0,0) {};
    	\node (u-label) at ($(u)+(360/9*6+90:0.4)$) {$u$};
    	\node[vertex] (u-1) at ($(u)+(360/9*7+90:\udist)$) {};
    	\node[vertex] (u-2) at ($(u)+(360/9*8+90:\udist)$) {};
    	\node[vertex] (u-3) at ($(u)+(360/9*1+90:\udist)$) {};
    	\node[vertex] (u-4) at ($(u)+(360/9*2+90:\udist)$) {};
    	\node[vertex] (u-5) at ($(u)+(360/9*4+90:\udist)$) {};
    	\node[vertex] (u-6) at ($(u)+(360/9*5+90:\udist)$) {};
    	
    	\node[vertex] (v) at ($(u-1)+(0:1.2*\vdist)$) {};
    	\node (v-label) at ($(v)+(225:0.3)$) {$v$};
    	\node (v-center) at ($(v)+(360:\vdist)$) {};
    	\node[vertex] (v-1) at ($(v-center)+(360/6*0:\vdist)$) {};
    	\node[vertex] (v-2) at ($(v-center)+(360/6*1:\vdist)$) {};
    	\node[vertex] (v-3) at ($(v-center)+(360/6*2:\vdist)$) {};
    	\node[vertex] (v-5) at ($(v-center)+(360/6*4:\vdist)$) {};
    	\node[vertex] (v-6) at ($(v-center)+(360/6*5:\vdist)$) {};
    	
    	\draw[directededge] (u) to (u-2);
    	\draw[directededge] (u-2) to (u-1);
    	\draw[directededge] (u-1) to (u);
    	\draw[directededge] (u) to (u-3);
    	\draw[directededge] (u-3) to (u-4);
    	\draw[directededge] (u-4) to (u);
    	\draw[directededge] (u) to (u-6);
    	\draw[directededge] (u-6) to (u-5);
    	\draw[directededge] (u-5) to (u);
    	
    	\draw[directededge] (u-1) to (v);
    	
    	\draw[directededge] (v) to (v-3);
    	\draw[directededge] (v-3) to (v-2);
    	\draw[directededge] (v-2) to (v-1);
    	\draw[directededge] (v-1) to (v-6);
    	\draw[directededge] (v-6) to (v-5);
    	\draw[directededge] (v-5) to (v);
    \end{tikzpicture}
  \end{subfigure}%
  \hfill%
  \begin{subfigure}[h]{0.49\textwidth}
    \centering
    \begin{tikzpicture}{scale=0.8}
    	\def\udist{2}
    	\def\vdist{0.8}
    	
    	\node[vertex] (u) at (0,0) {};
    	\node (u-label) at ($(u)+(360/9*6+90:0.4)$) {$u$};
    	\node[vertex] (u-1) at ($(u)+(360/9*7+90:\udist)$) {};
    	\node[vertex] (u-2) at ($(u)+(360/9*8+90:\udist)$) {};
    	\node[vertex] (u-3) at ($(u)+(360/9*1+90:\udist)$) {};
    	\node[vertex] (u-4) at ($(u)+(360/9*2+90:\udist)$) {};
    	\node[vertex] (u-5) at ($(u)+(360/9*4+90:\udist)$) {};
    	\node[vertex] (u-6) at ($(u)+(360/9*5+90:\udist)$) {};
    	
    	\node[vertex] (v) at ($(u-1)+(0:1.2*\vdist)$) {};
    	\node (v-label) at ($(v)+(225:0.3)$) {$v$};
    	\node (v-center) at ($(v)+(360:\vdist)$) {};
    	\node[vertex] (v-1) at ($(v-center)+(360/6*0:\vdist)$) {};
    	\node[vertex] (v-2) at ($(v-center)+(360/6*1:\vdist)$) {};
    	\node[vertex] (v-3) at ($(v-center)+(360/6*2:\vdist)$) {};
    	\node[vertex] (v-5) at ($(v-center)+(360/6*4:\vdist)$) {};
    	\node[vertex] (v-6) at ($(v-center)+(360/6*5:\vdist)$) {};
    	
    	\draw[directededge] (u) to (u-2);
    	\draw[directededge] (u-2) to (u-1);
    	\draw[directededge] (u-1) to (u);
    	\draw[directededge] (u) to (u-3);
    	\draw[directededge] (u-3) to (u-4);
    	\draw[directededge] (u-4) to (u);
    	\draw[directededge] (u) to (u-6);
    	\draw[directededge] (u-6) to (u-5);
    	\draw[directededge] (u-5) to (u);
    	
    	\draw[directededge] (u-1) to (v);
    	
    	\draw[directededge] (v) to (v-3);
    	\draw[directededge] (v-3) to (v-2);
    	\draw[directededge] (v-2) to (v-1);
    	\draw[directededge] (v-1) to (v-6);
    	\draw[directededge] (v-6) to (v-5);
    	\draw[directededge] (v-5) to (v);
    	
    	\draw (u) edge[directededge,myOrange,quick curve through={($(u)+(105:1)$) ($(u-3)+(90:\vdist)$) ($(u-4)+(180:\vdist)$) ($(u-5)+(270:\vdist)$) ($(u-6)+(270:\vdist)$) ($(v-1)+(0:\vdist)$) ($(v-2)+(90:\vdist)$) ($(v-3)+(90:0.6*\vdist)$) ($(v)+(75:1)$)}] (v);
    	\draw (u) edge[directededge,myOrange,quick curve through={($(u)+(70:1)$) ($(u-2)+(90:0.6*\vdist)$) ($(v)+(120:1)$)}] (v);
    	
    	\node[myOrange] (f-1) at ($(u-2)+(80:1)$) {$f_1$};
    	\node[myOrange] (f-2) at ($(u-6)+(290:1.2)$) {$f_2$};
    \end{tikzpicture}
  \end{subfigure}
  \caption{An~example plane digraph (left) and its $(u,v)$-simplification (right).}
  \label{fig:simplification-example}
\end{figure}
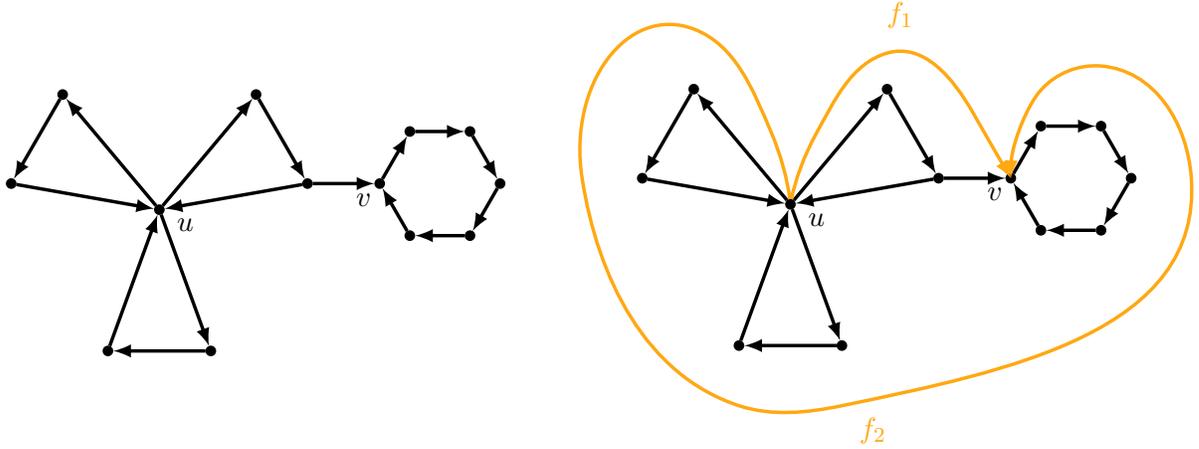

\paragraph*{Cutting arcs of plane graphs.}
We finally describe an~operation on plane graphs that introduces new vertices to the outer face $F$ of the embedding.
Assume that $F$ is simple.
Take two vertices $u \in V(F)$, $v \notin V(F)$, both incident to an~arc $e \in E(G)$.
An~operation of \emph{cutting the graph along} $e$ produces a~new graph $G_e$ which is the~result of the following process (see Figure~\ref{fig:arc-cutting}):
\begin{enumerate}
  \item Enumerate all arcs incident to $u$ in anti-clockwise order: $e_1, e_2, \dots, e_k$, where $e_1, e_k \in E(F)$.
  \item Let $\ell \in \{2, 3, \dots, k-1\}$ be such that $e_\ell = e$.
  \item Remove $u$ from the graph, along with all arcs incident to $u$.
  \item Introduce two new vertices $u_1$, $u_2$ to the graph.
  	For each $i \in \{1, 2, \dots, \ell\}$, add to the graph a~new arc $e_i^1$ which is obtained from the arc $e_i$ by replacing the endpoint $u$ with~$u_1$.
  	Similarly, for each $i \in \{\ell, \ell+1, \dots, k\}$, add to the graph a~new arc $e_i^2$ which is obtained from the arc $e_i$ by replacing the endpoint $u$ with $u_2$.
\end{enumerate}
Intuitively, $G_e$ is produced by drawing the graph $G$ on a~piece of paper and cutting the piece of paper along $e$.
Note that after this operation, the resulting graph is still planar (Figure~\ref{fig:arc-cutting}).
Moreover, the outer face $F_e$ of $G_e$ is simple, and $V(F) \setminus V(F_e) = \{u\}$ and $V(F_e) \setminus V(F) = \{u_1, v, u_2\}$.
That is, $v$ now lies on the outer face of the new graph.

We can generalize this procedure to paths instead of just edges.
Assume that $P = v_1 v_2 \dots v_k$ is a~$v_1$-grounded path which is disjoint with $V(F)$ except for $v_1$.
Then, cutting the graph along $P$ entails cutting the graph along the arcs $v_1v_2, v_2v_3, \dots, v_{k-1}v_k$, in this order.
This way we obtain a plane graph where the destination of the path lies on the outer face.
See Figure~\ref{fig:arc-cutting} for an illustration.

It is immediate that given the combinatorial embedding of a plane graph $G$, one can compute a plane embedding of the graph obtained by cutting $G$ along a path $P$ in time linear with respect to the size of the graph.
Thus we again obtain a plane graph.

\begin{figure}
\centering
  \begin{subfigure}[b]{0.3\textwidth}
    \centering
    \begin{tikzpicture}{scale=0.8}
    	\def\dist{1}
    	\definecolor{p-colour}{rgb}{1.0, 0.49, 0.0}
    	\definecolor{e-colour}{rgb}{0.55, 0.71, 0.0}
    	
    	\node[vertex,myGrey] (v-1) at (0,0) {};
    	\node[vertex,myGrey] (v-2) at ($(v-1)+(360:\dist)$) {};
    	\node[vertex,myGrey] (v-3) at ($(v-1)+(360/6*1:\dist)$) {};
    	\node[vertex,myGrey] (v-4) at ($(v-1)+(360/6*2:\dist)$) {};
    	\node (v) at ($(v-4)+(210:0.38)$) {\textcolor{e-colour}{$v$}};
    	\node[vertex,myGrey] (v-5) at ($(v-1)+(360/6*3:\dist)$) {};
    	\node[vertex,myGrey] (v-6) at ($(v-1)+(360/6*4:\dist)$) {};
    	\node[vertex,myGrey] (v-7) at ($(v-1)+(360/6*5:\dist)$) {};
    	\node[vertex,myGrey] (v-8) at ($(v-7)+(360:\dist)$) {};
    	\node[vertex,myGrey] (v-9) at ($(v-8)+(360/6*1:\dist)$) {};
    	\node[vertex,myGrey] (v-10) at ($(v-9)+(360/6*2:\dist)$) {};
    	\node[vertex,myGrey] (v-11) at ($(v-10)+(360/6*2:\dist)$) {};
    	\node[vertex,myGrey] (v-12) at ($(v-11)+(360/6*3:\dist)$) {};
    	\node (u) at ($(v-12)+(90:0.3)$) {\textcolor{e-colour}{$u$}};
    	\node[vertex,myGrey] (v-13) at ($(v-12)+(360/6*3:\dist)$) {};
    	\node[vertex,myGrey] (v-14) at ($(v-13)+(360/6*4:\dist)$) {};
    	\node[vertex,myGrey] (v-15) at ($(v-14)+(360/6*4:\dist)$) {};
    	\node[vertex,myGrey] (v-16) at ($(v-15)+(360/6*5:\dist)$) {};
    	\node[vertex,myGrey] (v-17) at ($(v-16)+(360/6*5:\dist)$) {};
    	\node[vertex,myGrey] (v-18) at ($(v-17)+(360:\dist)$) {};
    	\node[vertex,myGrey] (v-19) at ($(v-18)+(360:\dist)$) {};
    	
    	\draw[directededge,myGrey] (v-13) to (v-12);
    	\draw[directededge,myGrey] (v-12) to (v-11);
    	\draw[directededge,myGrey] (v-14) to (v-13);
    	\draw[directededge,myGrey] (v-4) to (v-13);
    	\draw[directededge,e-colour,ultra thick] (v-4) to (v-12);
    	\node (e) at ($(v-4)!0.5!(v-12)+(315:0.2)$) {\textcolor{e-colour}{$e$}};
    	\draw[directededge,myGrey] (v-12) to (v-3);
    	\draw[directededge,myGrey] (v-3) to (v-11);
    	\draw[directededge,myGrey] (v-11) to (v-10);
    	\draw[directededge,myGrey] (v-10) to (v-3);
    	\draw[directededge,myGrey] (v-3) to (v-4);
    	\draw[directededge,myGrey] (v-4) to (v-14);
    	\draw[directededge,myGrey] (v-14) to (v-15);
    	\draw[directededge,myGrey] (v-14) to (v-5);
    	\draw[directededge,myGrey] (v-5) to (v-1);
    	\draw[directededge,myGrey] (v-4) to (v-1);
    	\draw[directededge,myGrey] (v-3) to (v-1);
    	\draw[directededge,p-colour,ultra thick] (v-2) to (v-1);
    	\draw[directededge,myGrey] (v-2) to (v-3);
    	\draw[directededge,myGrey] (v-10) to (v-2);
    	\draw[directededge,myGrey] (v-9) to (v-10);
    	\draw[directededge,myGrey] (v-15) to (v-5);
    	\draw[directededge,myGrey] (v-9) to (v-2);
    	\draw[directededge,myGrey] (v-16) to (v-15);
    	\draw[directededge,myGrey] (v-5) to (v-16);
    	\draw[directededge,myGrey] (v-5) to (v-4);
    	\draw[directededge,p-colour,ultra thick] (v-6) to (v-5);
    	\draw[directededge,p-colour,ultra thick] (v-1) to (v-6);
    	\draw[directededge,myGrey] (v-7) to (v-1);
    	\draw[directededge,p-colour,ultra thick] (v-7) to (v-2);
    		\node (p) at ($(v-7)!0.5!(v-2)+(140:0.3)$) {\textcolor{p-colour}{$P$}};
    	\draw[directededge,myGrey] (v-2) to (v-8);
    	\draw[directededge,myGrey] (v-8) to (v-9);
    	\draw[directededge,myGrey] (v-8) to (v-7);
    	\draw[directededge,myGrey] (v-7) to (v-6);
    	\draw[directededge,myGrey] (v-6) to (v-16);
    	\draw[directededge,myGrey] (v-17) to (v-16);
    	\draw[directededge,myGrey] (v-18) to (v-17);
    	\draw[directededge,myGrey] (v-6) to (v-18);
    	\draw[directededge,p-colour,ultra thick] (v-18) to (v-7);
    	\draw[directededge,myGrey] (v-19) to (v-18);
    	\draw[directededge,myGrey] (v-17) to (v-6);
    	\draw[directededge,myGrey] (v-7) to (v-19);
    	\draw[directededge,myGrey] (v-19) to (v-8);
    \end{tikzpicture}
    \caption{}
  \end{subfigure}%
  \hfill%
  \begin{subfigure}[b]{0.3\textwidth}
    \centering
    \begin{tikzpicture}{scale=0.8}
    	\def\dist{1}
    	\definecolor{p-colour}{rgb}{1.0, 0.49, 0.0}
    	\definecolor{e-colour}{rgb}{0.55, 0.71, 0.0}
    	
    	\node[vertex,myGrey] (v-1) at (0,0) {};
    	\node[vertex,myGrey] (v-2) at ($(v-1)+(360:\dist)$) {};
    	\node[vertex,myGrey] (v-3) at ($(v-1)+(360/6*1:\dist)$) {};
    	\node[vertex,myGrey] (v-4) at ($(v-1)+(360/6*2:\dist)$) {};
    	\node (v) at ($(v-4)+(210:0.38)$) {\textcolor{e-colour}{$v$}};
    	\node[vertex,myGrey] (v-5) at ($(v-1)+(360/6*3:\dist)$) {};
    	\node[vertex,myGrey] (v-6) at ($(v-1)+(360/6*4:\dist)$) {};
    	\node[vertex,myGrey] (v-7) at ($(v-1)+(360/6*5:\dist)$) {};
    	\node[vertex,myGrey] (v-8) at ($(v-7)+(360:\dist)$) {};
    	\node[vertex,myGrey] (v-9) at ($(v-8)+(360/6*1:\dist)$) {};
    	\node[vertex,myGrey] (v-10) at ($(v-9)+(360/6*2:\dist)$) {};
    	\node[vertex,myGrey] (v-11) at ($(v-10)+(360/6*2:\dist)$) {};
    	\node (v-12) at ($(v-11)+(360/6*3:\dist)$) {};
    	\node[vertex,myGrey] (v-12-a) at ($(v-12)+(180:0.2)$) {};
    	\node[vertex,myGrey] (v-12-b) at ($(v-12)+(0:0.2)$) {};
    	\node (u-a) at ($(v-12-a)+(90:0.3)$) {\textcolor{e-colour}{$u_2$}};
    	\node (u-b) at ($(v-12-b)+(90:0.3)$) {\textcolor{e-colour}{$u_1$}};
    	\node[vertex,myGrey] (v-13) at ($(v-12)+(360/6*3:\dist)$) {};
    	\node[vertex,myGrey] (v-14) at ($(v-13)+(360/6*4:\dist)$) {};
    	\node[vertex,myGrey] (v-15) at ($(v-14)+(360/6*4:\dist)$) {};
    	\node[vertex,myGrey] (v-16) at ($(v-15)+(360/6*5:\dist)$) {};
    	\node[vertex,myGrey] (v-17) at ($(v-16)+(360/6*5:\dist)$) {};
    	\node[vertex,myGrey] (v-18) at ($(v-17)+(360:\dist)$) {};
    	\node[vertex,myGrey] (v-19) at ($(v-18)+(360:\dist)$) {};
    	
    	\draw[directededge,myGrey] (v-13) to (v-12-a);
    	\draw[directededge,myGrey] (v-12-b) to (v-11);
    	\draw[directededge,myGrey] (v-14) to (v-13);
    	\draw[directededge,myGrey] (v-4) to (v-13);
    	\draw[directededge,e-colour,ultra thick] (v-4) to (v-12-a);
    	\draw[directededge,e-colour,ultra thick] (v-4) to (v-12-b);
    	\draw[directededge,myGrey] (v-12-b) to (v-3);
    	\draw[directededge,myGrey] (v-3) to (v-11);
    	\draw[directededge,myGrey] (v-11) to (v-10);
    	\draw[directededge,myGrey] (v-10) to (v-3);
    	\draw[directededge,myGrey] (v-3) to (v-4);
    	\draw[directededge,myGrey] (v-4) to (v-14);
    	\draw[directededge,myGrey] (v-14) to (v-15);
    	\draw[directededge,myGrey] (v-14) to (v-5);
    	\draw[directededge,myGrey] (v-5) to (v-1);
    	\draw[directededge,myGrey] (v-4) to (v-1);
    	\draw[directededge,myGrey] (v-3) to (v-1);
    	\draw[directededge,myGrey] (v-2) to (v-1);
    	\draw[directededge,myGrey] (v-2) to (v-3);
    	\draw[directededge,myGrey] (v-10) to (v-2);
    	\draw[directededge,myGrey] (v-9) to (v-10);
    	\draw[directededge,myGrey] (v-15) to (v-5);
    	\draw[directededge,myGrey] (v-9) to (v-2);
    	\draw[directededge,myGrey] (v-16) to (v-15);
    	\draw[directededge,myGrey] (v-5) to (v-16);
    	\draw[directededge,myGrey] (v-5) to (v-4);
    	\draw[directededge,myGrey] (v-6) to (v-5);
    	\draw[directededge,myGrey] (v-1) to (v-6);
    	\draw[directededge,myGrey] (v-7) to (v-1);
    	\draw[directededge,myGrey] (v-7) to (v-2);
    	\draw[directededge,myGrey] (v-2) to (v-8);
    	\draw[directededge,myGrey] (v-8) to (v-9);
    	\draw[directededge,myGrey] (v-8) to (v-7);
    	\draw[directededge,myGrey] (v-7) to (v-6);
    	\draw[directededge,myGrey] (v-6) to (v-16);
    	\draw[directededge,myGrey] (v-17) to (v-16);
    	\draw[directededge,myGrey] (v-18) to (v-17);
    	\draw[directededge,myGrey] (v-6) to (v-18);
    	\draw[directededge,myGrey] (v-18) to (v-7);
    	\draw[directededge,myGrey] (v-19) to (v-18);
    	\draw[directededge,myGrey] (v-17) to (v-6);
    	\draw[directededge,myGrey] (v-7) to (v-19);
    	\draw[directededge,myGrey] (v-19) to (v-8);
    \end{tikzpicture}
    \caption{}
  \end{subfigure}%
  \hfill%
  \begin{subfigure}[b]{0.3\textwidth}
    \centering
    \begin{tikzpicture}{scale=0.8}
    	\def\dist{1}
    	\definecolor{p-colour}{rgb}{1.0, 0.49, 0.0}
    	\definecolor{e-colour}{rgb}{0.55, 0.71, 0.0}
    	
    	\node (v-1) at (0,0) {};
    	\node[vertex,myGrey] (v-1-a) at ($(v-1)+(120:0.1)$) {};
    	\node[vertex,myGrey] (v-1-b) at ($(v-1)+(315:0.2)$) {};
    	\node (v-2) at ($(v-1)+(360:\dist)$) {};
    	\node[vertex,myGrey] (v-2-a) at ($(v-2)+(35:0.25)$) {};
    	\node[vertex,myGrey] (v-2-b) at ($(v-2)+(210:0.25)$) {};
    	\node[vertex,myGrey] (v-3) at ($(v-1)+(360/6*1:\dist)$) {};
    	\node[vertex,myGrey] (v-4) at ($(v-1)+(360/6*2:\dist)$) {};
    	\node[vertex,myGrey] (v-5) at ($(v-1)+(360/6*3:\dist)$) {};
    	\node (v-6) at ($(v-1)+(360/6*4:\dist)$) {};
    	\node[vertex,myGrey] (v-6-a) at ($(v-6)+(90:0.3)$) {};
    	\node[vertex,myGrey] (v-6-b) at ($(v-6)+(270:0.2)$) {};
    	\node (v-7) at ($(v-1)+(360/6*5:\dist)$) {};
    	\node[vertex,myGrey] (v-7-a) at ($(v-7)+(325:0.18)$) {};
    	\node[vertex,myGrey] (v-7-b) at ($(v-7)+(140:0.18)$) {};
    	\node[vertex,myGrey] (v-8) at ($(v-7)+(360:\dist)$) {};
    	\node[vertex,myGrey] (v-9) at ($(v-8)+(360/6*1:\dist)$) {};
    	\node[vertex,myGrey] (v-10) at ($(v-9)+(360/6*2:\dist)$) {};
    	\node[vertex,myGrey] (v-11) at ($(v-10)+(360/6*2:\dist)$) {};
    	\node[vertex,myGrey] (v-12) at ($(v-11)+(360/6*3:\dist)$) {};
    	\node[vertex,myGrey] (v-13) at ($(v-12)+(360/6*3:\dist)$) {};
    	\node[vertex,myGrey] (v-14) at ($(v-13)+(360/6*4:\dist)$) {};
    	\node[vertex,myGrey] (v-15) at ($(v-14)+(360/6*4:\dist)$) {};
    	\node[vertex,myGrey] (v-16) at ($(v-15)+(360/6*5:\dist)$) {};
    	\node[vertex,myGrey] (v-17) at ($(v-16)+(360/6*5:\dist)$) {};
    	\node (v-18) at ($(v-17)+(360:\dist)$) {};
    	\node[vertex,myGrey] (v-18-a) at ($(v-18)+(0:0.2)$) {};
    	\node[vertex,myGrey] (v-18-b) at ($(v-18)+(180:0.2)$) {};
    	\node[vertex,myGrey] (v-19) at ($(v-18)+(360:\dist)$) {};
    	
    	\draw[directededge,myGrey] (v-13) to (v-12);
    	\draw[directededge,myGrey] (v-12) to (v-11);
    	\draw[directededge,myGrey] (v-14) to (v-13);
    	\draw[directededge,myGrey] (v-4) to (v-13);
    	\draw[directededge,myGrey] (v-4) to (v-12);
    	\draw[directededge,myGrey] (v-12) to (v-3);
    	\draw[directededge,myGrey] (v-3) to (v-11);
    	\draw[directededge,myGrey] (v-11) to (v-10);
    	\draw[directededge,myGrey] (v-10) to (v-3);
    	\draw[directededge,myGrey] (v-3) to (v-4);
    	\draw[directededge,myGrey] (v-4) to (v-14);
    	\draw[directededge,myGrey] (v-14) to (v-15);
    	\draw[directededge,myGrey] (v-14) to (v-5);
    	\draw[directededge,myGrey] (v-5) to (v-1-a);
    	\draw[directededge,myGrey] (v-4) to (v-1-a);
    	\draw[directededge,myGrey] (v-3) to (v-1-a);
    	\draw[directededge,p-colour,ultra thick] (v-2-a) to (v-1-a);
    	\draw[directededge,p-colour,ultra thick] (v-2-b) to (v-1-b);
    	\draw[directededge,myGrey] (v-2-a) to (v-3);
    	\draw[directededge,myGrey] (v-10) to (v-2-a);
    	\draw[directededge,myGrey] (v-9) to (v-10);
    	\draw[directededge,myGrey] (v-15) to (v-5);
    	\draw[directededge,myGrey] (v-9) to (v-2-a);
    	\draw[directededge,myGrey] (v-16) to (v-15);
    	\draw[directededge,myGrey] (v-5) to (v-16);
    	\draw[directededge,myGrey] (v-5) to (v-4);
    	\draw[directededge,p-colour,ultra thick] (v-6-a) to (v-5);
    	\draw[directededge,p-colour,ultra thick] (v-6-b) to (v-5);
    	\draw[directededge,p-colour,ultra thick] (v-1-a) to (v-6-a);
    	\draw[directededge,p-colour,ultra thick] (v-1-b) to (v-6-b);
    	\draw[directededge,myGrey] (v-7-b) to (v-1-b);
    	\draw[directededge,p-colour,ultra thick] (v-7-a) to (v-2-a);
    	\draw[directededge,p-colour,ultra thick] (v-7-b) to (v-2-b);
    	\draw[directededge,myGrey] (v-2-a) to (v-8);
    	\draw[directededge,myGrey] (v-8) to (v-9);
    	\draw[directededge,myGrey] (v-8) to (v-7-a);
    	\draw[directededge,myGrey] (v-7-b) to (v-6-b);
    	\draw[directededge,myGrey] (v-6-b) to (v-16);
    	\draw[directededge,myGrey] (v-17) to (v-16);
    	\draw[directededge,myGrey] (v-18-b) to (v-17);
    	\draw[directededge,myGrey] (v-6-b) to (v-18-b);
    	\draw[directededge,p-colour,ultra thick] (v-18-a) to (v-7-a);
    	\draw[directededge,p-colour,ultra thick] (v-18-b) to (v-7-b);
    	\draw[directededge,myGrey] (v-19) to (v-18-a);
    	\draw[directededge,myGrey] (v-17) to (v-6-b);
    	\draw[directededge,myGrey] (v-7-a) to (v-19);
    	\draw[directededge,myGrey] (v-19) to (v-8);
    \end{tikzpicture}
    \caption{}
  \end{subfigure}
  \caption{(a) A~directed plane graph $G$. (b) The result of cutting $G$ along $e$. (c) The result of cutting $G$ along $P$.}
  \label{fig:arc-cutting}
\end{figure}
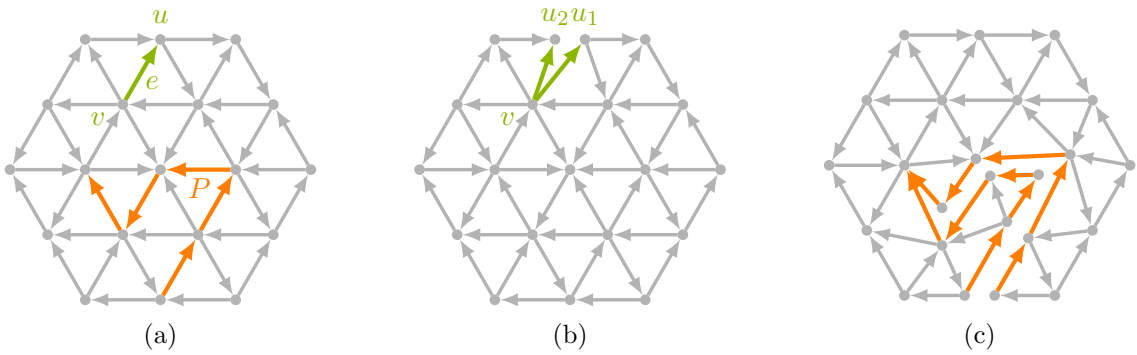

\section{Directed Detour}\label{sec:detour-algo}

In this section we prove Theorem~\ref{thm:main-ddetour}.

Let $G$ be a~plane digraph, and let $s, t \in V(G)$ be a pair of vertices. Recall that we assume $G$ to be weakly connected and we fix the plane embedding of $G$. Therefore, from now on for simplicity we identify features in $G$ (vertices, edges, paths, etc.) with their images under the embedding.
Our goal is to decide whether there is an $s$-to-$t$ path in $G$ of length at least $\dist{G}{s}{t} + 1$.
We are going to reduce this question to solving a set of instances of the \dpaths{2} problem.
The reduction is similar to the algorithm of Fomin et al.~\cite{DetoursDirected} for \dldetour.

First, we run breadth-first search (BFS) starting from $s$ in $G$.
Let $L_i$ denote the $i$-th layer of this BFS, that is, $L_i = \{ v \in V(G) \mid \dist{G}{s}{v} = i \}$ for $i \in \Set{0, 1, \ldots, n}$.

Suppose now that the instance $(G, s, t)$ is a~yes-instance, and let $P$ be an~$s$-to-$t$ path witnessing this fact.
Since $P$ is a~non-shortest path from $s$ to $t$, there is an~index $p\in \{0,1,\ldots,n\}$ such that the layer $L_p$ contains at least two vertices of $P$.
Let us choose $p$ to be the smallest such an~index.
Let $G_{\geq p}$ be the~plane digraph obtained from $G$ by removing all the vertices in $L_0 \cup L_1 \cup \ldots \cup L_{p-1}$.
We may assume that $\Gp$ is weakly connected, for otherwise we discard all of its weakly connected components that do not contain the vertex~$t$.

\begin{claim}\label{claim:one-face}
  All vertices of $L_p$ lie on one face of~$G_{\geq p}$.
\end{claim}

\begin{proof}
  Consider any vertex $v$ of $L_p$ and any shortest path $Q$ from $s$ to $v$ in $G$. Observe that all vertices of $Q$ except for $v$ lie in layers $L_0,L_1,\ldots,L_{p-1}$, hence they are removed when constructing $G_{\geq p}$ from $G$. We conclude that $v$ lies on the boundary of the (unique) face of $G_{\geq p}$ that contains $s$.
\end{proof}

Denote by $x$ and $y$, respectively, the first and the second vertex on the path $P$ that lie in the layer $L_p$.
In our algorithm we iterate over all possible choices for the vertex~$y$.
Note that the choice of $y$ determines the value of $p$, because $y \in L_p$.
Observe that if we guess the vertex $y$ correctly, then in order to find a~non-shortest path from $s$ to $t$ it is enough to find a vertex $x \in L_p$ ($x \neq y$) and three paths $\pstart$, $\pmid$ and $\pend$ such that
\begin{itemize}
 \item $\pstart$ is a~shortest path from $s$ to $x$ in $G$; and
 \item $\pmid$, $\pend$ are two internally vertex-disjoint paths in $G_{\geq p}$ going respectively from $x$ to $y$ and from $y$ to $t$.
\end{itemize}
Note that $y$ might be equal to $t$ and then $\pend$ is a trivial path only consisting of the single vertex $t$.

On one hand, vertices $x$ and $y$ divide $P$ into subpaths $\pstart$, $\pmid$, and $\pend$ satisfying the properties stated above. On the other hand, if we find a vertex $x$ and paths $\pstart,\pmid,\pend$ satisfying the above, then their concatenation forms a~valid solution. This is because the path $\pstart$ goes consecutively through layers $L_0, L_1, \ldots, L_p$ due to being a shortest path, and thus it is internally vertex-disjoint from $\pmid$ and $\pend$. Also, the concatenation of $\pstart,\pmid,\pend$ is a non-shortest path from $s$ to $t$ due to containing at least two different vertices from layer $L_p$.

It would be natural to also iterate through all possible choices for $x$, but for the sake of optimizing the running time we do the following instead.
Construct a~digraph $H$  from the digraph $G_{\geq p}$ by adding to it a~vertex $\xsuper$ together with arcs $\xsuper u$ for all $u \in L_p\setminus \{y\}$.

\begin{claim}\label{claim:super-vertex}
  $H$ is a~planar digraph and its implicit embedding can be obtained from the~implicit embedding of $G_{\geq p}$ in linear time.
\end{claim}

\begin{proof}
  It suffices to extend the implicit embedding of $G_{\geq p}$ by embedding $\xsuper$ anywhere in the unique face whose boundary contains all vertices of $L_p$ (which exists by Claim~\ref{claim:one-face}) and draw arcs connecting $\xsuper$ with vertices of $L_p\setminus \{y\}$ within this face. It is straightforward to see that this can be done in the setting of implicit embeddings in linear time.
\end{proof}

Observe that pairs $(P_1, P_2)$ of internally vertex-disjoint paths going respectively from $\xsuper$ to $y$ in $H$ and from $y$ to $t$ in $H$ correspond one-to-one to pairs $(\pmid, \pend)$ of internally vertex-disjoint paths going respectively from some vertex $x \in L_p$ ($x \neq y$) to $y$ in $G_{\geq p}$ and from $y$ to $t$ in $G_{\geq p}$.
We may additionally assume that the face of $H$ containing both $\xsuper$ and $y$ is the outer face of $H$.
This way we reduced our original problem to a set of at most $n$ instances (one instance per each choice of the vertex $y$) of the following problem:

\begin{tcolorbox}
\noindent \proba

\smallskip

\noindent\textbf{Input:} A~plane digraph~$H$ and three vertices $x, y, t \in V(H)$, where $x$ and $y$ lie on the outer face of $H$.

\smallskip

\noindent\textbf{Question:} Are there two internally vertex-disjoint paths $P_1$ and $P_2$ such that $P_1$ is an~$x$-to-$y$ path and $P_2$ is a~$y$-to-$t$ path in $H$?
\end{tcolorbox}

So to conclude Theorem~\ref{thm:main-ddetour} it is enough to show the following lemma.

\begin{lemma}
  \proba can be solved in time $\Oh(n)$.
\end{lemma}

\begin{proof}
  Let $H^\star$ be a~$(y,x)$-simplification of the outer face of $H$.
  Naturally, the outer face of $H^\star$ is simple.
  Moreover, $(H^\star, x, y, t)$ is a~yes-instance if and only if $(H, x, y, t)$ is a~yes-instance: $H^\star$ is exactly the graph $H$ with two additional $yx$ arcs, which plainly cannot be a~part of the sought solution.

  Suppose now that $(H^\star, x, y, t)$ is a~yes-instance, and let $(P_1, P_2)$ be a~pair of paths in~$H^\star$ that form a solution to \proba on $H^\star$.
  $P_1$ connects $x$ to $y$ which lie on the outer face $F_O$ of $H^\star$ (which is a simple cycle), thus it is $\Brace{x,y}$-grounded with respect to $F_O$.
  Therefore, it splits the whole digraph~$H$ into two parts, the area $\leftarea{P_1}$ left of $P_1$ and the area $\rightarea{P_1}$ right of $P_1$, with both areas intersecting only at $P_1$.
  Since $P_2$ is internally disjoint with $P_1$, it must be entirely contained within $\leftarea{P_1}$ or $\rightarea{P_1}$; without loss of generality assume that $P_2$ is contained in the latter.
  Then, let $L_{x,y}$ be the leftmost path from $x$ to $y$ in $H^\star$.

  \begin{claim}
    $(L_{x,y}, P_2)$ is also a~valid solution to \proba.
  \end{claim}

  \begin{proof}
    By the definition of the leftmost path we know that $\leftarea{L_{x,y}} \subseteq \leftarea{P_1}$, that is, $L_{x,y}$ cannot use any vertex lying on the right side of $P_1$.
    Since all vertices of $P_2$ (apart from $y$) lie on the right side of $P_1$ we conclude the claim.
  \end{proof}
  
  An~analogous argument shows that if $P_2$ is contained in $\leftarea{P_1}$, then $(R_{x,y}, P_2)$ is a~valid solution to \proba, where $R_{x,y}$ is the rightmost path from $x$ to $y$ in $H^\star$.
  By the observation above, we know that in order to verify whether $(H^\star, x, y, t)$ is a~yes-instance it is enough to find $L_{x,y}$ and $R_{x,y}$ as candidates for $P_1$, and to check for each candidate whether the vertex $t$ is reachable from $y$ in $H^\star - (V(P_1) \setminus \{y\})$. 
  All these checks can be done in time $\Oh(n)$ by running a proper depth-first search algorithm.
\end{proof}

\section{Directed Long Detour}\label{sec:long-detour-algo}

In this section we prove Theorem~\ref{thm:main-dldetour}.

We are given a~plane digraph $G$, two vertices $s, t \in V(G)$ and an~integer $k \in \mathds{N}$.
As in Section~\ref{sec:detour-algo}, we assume that $G$ is weakly connected, and we fix the plane embedding of $G$.
We need to decide whether there is an~$s$-to-$t$ path in $G$ of length at least $\dist{G}{s}{t} + k$.
We are going to reduce this question to solving a set of instances of the \dpaths{2} problem (with some additional requirements).

We begin with running the algorithm of Bez\'akov\'a et al.~\cite{BezakovaCDF19} to check whether there is an~$s$-to-$t$ path of length $\dist{G}{s}{t} + l$ for some $l \in \Set{k, k + 1, \ldots, 3k - 1}$.
From now on, we assume there is no such $s$-to-$t$ path.

As in Section~\ref{sec:detour-algo}, we run a BFS starting from $s$ in $G$ and set $L_i = \{ v \in V(G) \mid \dist{G}{s}{v} = i\}$ for $i = 0, 1, \ldots, n$.
For $i = 1, 2, \ldots, n$, let $G_{\geq i}$ be the plane digraph obtained from $G$ by removing all the vertices in $L_0 \cup \ldots \cup L_{i-1}$.
Again, without loss of generality, we assume that $G_{\geq i}$ is weakly connected.

Now, we iterate through all values for $p = 1, 2, \ldots, n$, and through all choices for $x, y \in L_p$, where $x \neq y$ (there are at most $n^2$ choices for $(p, x, y)$).
Recall from Section~\ref{sec:detour-algo} (Section~\ref{claim:one-face}) that both $x$ and $y$ lie on the outer face of $\Gp$. 
Next, let $\Gp^{x,y}$ be the $(y,x)$-simplification of the outer face of $\Gp$.
We are now going to search for three paths $\pstart$, $\pmid$ and $\pend$ satisfying the conditions:
\begin{enumerate}
  \item[$(a)$] $\pstart$ is a~shortest $s$-to-$x$ path in $G$;
  \item[$(b)$] $\pmid$ is an~$(x,y)$-grounded path in $\Gp^{x,y}$ of length at least $2k$;
  \item[$(c)$] $\pend$ is a~$y$-to-$t$ path in $\Gp^{x,y}$;
  \item[$(d)$] paths $\pmid$ and $\pend$ are internally vertex-disjoint; and
  \item[$(e)$] if $\pend \subseteq \rightarea{\pmid}$, then there is no $(x, y)$-grounded path $\pmid'$ in $\Gp^{x,y}$ of length at least $k$ such that $\pmid' \lefttop \pmid$; and if $\pend \subseteq \leftarea{\pmid}$, then there is no such path with $\pmid \lefttop \pmid'$.
\end{enumerate}

We remark that the two $yx$ arcs added in the process of $(y,x)$-simplification of $\Gp$ cannot be a~part of any of the paths $\pstart$, $\pmid$, $\pend$ and hence their existence can be safely ignored in the following series of claims.

First, we show that the procedure described above is actually equivalent to solving the \dldetour problem on the instance $(G, s, t, k)$.

\begin{claim}\label{claim:concat-solution}
  Assume that the paths $\pstart$, $\pmid$ and $\pend$ satisfy the properties $(a)$ -- $(d)$ above.
  Then the concatenation
  \[ P = \pstart \circ \pmid \circ \pend \]
  forms a valid solution for the \dldetour problem.
\end{claim}

\begin{proof}
First, $P$ is a~simple $s$-to-$t$ path.
Indeed, $\pmid$ and $\pend$ are internally vertex-disjoint by property $(d)$.
Moreover, since the path $\pstart$ is a shortest $s$-to-$x$ path in $G$, it does not use any vertex of $L_p \cup \ldots \cup L_n = V(\Gp)$ apart from $x$, and thus $\pstart$ is internally vertex-disjoint from $\pmid$ and $\pend$ as well.

By property $(b)$ the length of $P$ is at least
\[
\dist{G}{s}{x} + 2k + \dist{G}{y}{t} = \dist{G}{s}{y} + 2k + \dist{G}{y}{t} \geq \dist{G}{s}{t} + 2k,
\]
which finishes the proof.
\end{proof}

\begin{claim}
  If $(G, s, t, k)$ is a~yes-instance, then there exist an integer $p \in \mathds{N}$ and vertices $x, y \in L_p$ for which there are paths $\pstart$, $\pmid$ and $\pend$ satisfying the properties $(a)$ -- $(e)$.
\end{claim}

\begin{proof}
Assume that $(G, s, t, k)$ is a yes-instance.
Let $P$ be an~$s$-to-$t$ path witnessing this fact.
If there are many such paths, we choose $P$ to be a~shortest one.
As we established that there is no $s$-to-$t$ path of length $\dist{G}{s}{t} + l$ for any $l \in \Set{k, k + 1, \ldots, 3k - 1}$, we may assume that the path~$P$ is of length at least $\dist{G}{s}{t} + 3k$.
Since $P$ is a non-shortest $s$-to-$t$ path, we may define $p \in \{1, 2, \ldots, n\}$ to be the smallest index such that $L_p$ contains at least two vertices of the path $P$.
Let $x$ and $y$ be, respectively, the first and the second vertex on the path $P$ that lie in the layer~$L_p$.

By definition of $p$, the subpath $P[s \to x]$ is a shortest $s$-to-$x$ path, and thus we may set $\pstart \coloneqq P[s \to x]$.
We also set $\pend \coloneqq P[y \to t]$.
We will choose $\pmid$ later in the course of the proof.

Let us observe that the length of the subpath $P[x \to y]$ is at least $2k$: otherwise, we consider the path $P'$ being the concatenation of a shortest $s$-to-$y$ path in $G$ and the subpath $P[y \to t]$.
As in Claim~\ref{claim:concat-solution}, we argue that $P'$ is a simple $s$-to-$t$ path in $G$, and the length of $P'$ is
\begin{align*}
\dist{G}{s}{y} + \length{P[y \to t]} & = \dist{G}{s}{x} + \length{P[y \to t]} \\
& = \length{P} - \length{P[x \to y]} \\
& > (\dist{G}{s}{t} + 3k) - 2k = \dist{G}{s}{t} + k.
\end{align*}
Hence, $P'$ is also a valid solution for our instance $(G, s, t, k)$ that is shorter than $P$, which is a contradiction.

Since $P[x \to y]$ is an~$(x, y)$-grounded path in $\Gp^{x,y}$ and is internally disjoint from $P[y \to t]$, we may assume that $P[y \to t] \subseteq \rightarea{P[x \to y]}$ in $\Gp^{x,y}$; the case $P[y \to t] \subseteq \leftarea{P[x \to y]}$ is symmetric.

Observe now, that if we set $\pmid \coloneqq P[x \to y]$, then the properties $(a)$ -- $(d)$ are satisfied.
If the condition $(e)$ holds as well, then we are done.
Otherwise, we define $\pmid$ as a~minimal, with respect to $\lefttop$, $(x,y)$-grounded path in $\Gp^{x,y}$ of length at least $2k$.
Then, $\pmid \lefttop P[x \to y]$, and consequently $\pmid$ is internally vertex disjoint from $\pend = P[y \to t]$ as we have $P[y \to t] \subseteq \rightarea{P[x \to y]}$.

It remains to show that the pair $(\pmid, \pend)$ satisfies the property $(e)$.
By the definition of $\pmid$ we know there is no $(x,y)$-grounded path of length at least $2k$ which is topologically left of $\pmid$.
Suppose now that there exists an~$(x,y)$-grounded path $\pmid'$ in $\Gp^{x,y}$ such that $\pmid' \lefttop \pmid$ and $\length{\pmid'} \in [k, 2k - 1]$.
Consider the following walk $P'$:
\[ P' = P[s \to x] \circ \pmid' \circ P[y \to t]. \]
$P'$ is a simple path as in $\Gp$ we have $\pmid' \lefttop \pmid \lefttop P[x \to y]$ and $P[y \to t] \subseteq \rightarea{P[x \to y]}$, therefore $\pmid'$ is internally vertex-disjoint from both $P[s \to x]$ and $P[y \to t]$.
Moreover, the length of $P'$ is at least
\begin{align*}
\dist{G}{s}{x} + \length{\pmid'} + \dist{G}{y}{t} =& \\
\dist{G}{s}{y} + \dist{G}{y}{t} + \length{\pmid'} \geq&\ \dist{G}{s}{t} + \length{\pmid'} \geq \dist{G}{s}{t} + k.
\end{align*}
Consequently, $P'$ is a valid solution for the instance $(G, s, t, k)$.
Finally, since $\length{\pmid'} < 2k \leq \length{P[x \to y]}$, the path $P'$ is shorter than the path $P$ which is a contradiction.

\end{proof}

It remains to show how we can find paths $\pstart$, $\pmid$ and $\pend$ satisfying the above-mentioned conditions.
Let us assume that we guess the values of $p$, $x$ and $y$ correctly.
By Claim~\ref{claim:concat-solution}, we know that $\pstart$ can be set to any shortest $s$-to-$x$ path, and we see that in order to find the desired paths $\pmid$ and $\pend$ we can restrict ourselves to the digraph $\Gp^{x,y}$.
Let us call a pair of paths $(\pmid, \pend)$ \emph{special} for $(\Gp^{x,y}, x, y, t, k)$ if they satisfy the properties $(b)$ -- $(e)$ stated above.
This reduces our instance of the \dldetour problem to solving the set of at most $n^2$ instances (one for each choice of $x$ and $y$) of the following problem (with $H \coloneqq \Gp^{x,y}$).

\begin{tcolorbox}
\noindent \probb

\smallskip

\noindent\textbf{Input:} A~plane digraph~$H$, three vertices $x, y, t \in V(H)$, and an integer $k \in \mathds{N}$, where $x$ and $y$ lie on the outer face of $H$.
The outer face of $H$ is simple and contains only the vertices $x$ and $y$.

\smallskip

\noindent\textbf{Question:} Does there exist a~special pair of paths $(P_1, P_2)$ for $(H, x, y, t, k)$?

\end{tcolorbox}

To finish the proof of Claim~\ref{thm:main-dldetour} it is enough to show the following lemma.

\begin{lemma}
  \probb can be solved in time $2^{\Oh(k)} \cdot n^2$ by a Monte Carlo algorithm, or deterministically in time $2^{\Oh(k)} \cdot n^2 \log n$.
\end{lemma}

\begin{proof}
We use a variant of the standard color coding technique \cite{ColorCoding}.
Let us color independently every vertex of $H$ with one of two colors, say green and blue, each with probability $\frac{1}{2}$.

Assume for now that $(H, x, y, t, k)$ is a yes-instance, and let $(P_1, P_2)$ be a special pair of paths for our instance.
As per our previous considerations, we know that $P_2 \subseteq \leftarea{P_1}$ or $P_2 \subseteq \rightarea{P_1}$.
Without loss of generality assume that $P_2 \subseteq \rightarea{P_1}$.
Recall that $\length{P_1} \geq 2k$.
Therefore, we may define $x'$ to be the $k$-th vertex on $P_1$ and $y'$ to be the $k$-th vertex from the end of $P_1$.
Then, $\length{P_1[x \to x']} = \length{P_1[y' \to y]} = k-1$ and the paths $P_1[x \to x']$ and $P_1[y' \to y]$ are vertex-disjoint.

\newcommand{\greenset}{\Gc}

Suppose that the colors are assigned in such a way that all vertices of the path $P_1[x \to x']$ are green, and all vertices of the path $P_1[y' \to y]$ are blue.
The probability of such an event occurring is $(1/2)^{2k}$.
Let $\greenset$ be the set of green vertices of $H$ in the random coloring.
Note that we assume that $V(P_1[x \to x']) \subseteq \greenset$.


\begin{claim}\label{claim:leftmost-green}
$P_1[x \to x']$ is the lexicographically leftmost $x$-to-$x'$ path in $H$ entirely contained in~$\greenset$.
\end{claim}

\begin{proof}
Suppose that there is an~$x$-to-$x'$ path $Q$ in $H$, entirely contained in~$\greenset$, such that $Q \leftlex P_1[x \to x']$.
Let $q \in V(H)$ be the vertex for which path $P_1[x \to q]$ is the longest common prefix of $P_1$ and $Q$.
Also, let $r$ be the first vertex on $Q[q \to x']$, not including $q$, such that $r \in V(P_1)$.
Here, $r$ is well-defined as the vertex $x'$ belongs to both $Q[q \to x']$ and $P_1$.
Since $r$ lies on $Q$, we necessarily have that $r$ is green.
Thus, $r \not\in V(P_1[y' \to y])$ since $P_1[y' \to y]$ consists only of blue vertices.
Hence, the following walk $P_1'$:
\[ P_1' = Q[x \to r] \circ P_1[r \to y] \]
is an~$(x, y)$-grounded path in $H$ of length at least
\[
 \length{P_1[y' \to y])} + 1 = k.
\]
Recall that $Q \leftlex P_1$ in $H$.
By the definition of $r$ we observe that $Q[x \to r] \lefttop P_1[x \to r]$.
Hence, $P_1' \lefttop P_1$.
We conclude that the path $P_1'$ has length at least $k$ and is left of $P_1$, which contradicts the property $(e)$ of the pair $(P_1, P_2)$.
\end{proof}

\newcommand{\Hsplit}{H^\textrm{cut}}

Note that the path $P_1[x \to x']$ is disjoint from the set of vertices on the outer face of $H$, apart from its first vertex $x$.
This allows us to define $\Hsplit$ as the plane digraph obtained from $H$ by cutting along the path $P_1[x \to x']$.
Next, let $Z \subseteq V(\Hsplit)$ be the set of new vertices introduced to $H'$ by this operation.
Let $H'$ be a~$(y,x')$-simplification of the outer face of $\Hsplit$.

\begin{claim}
$P_1[x' \to y]$ is the leftmost path in $H'$ among all the paths between $x'$ and $y$ which do not contain any vertex of $Z$.
\end{claim}

\begin{proof}
Suppose that there is an~$x'$-to-$y$ path $Q$ in $H'$ which avoids $Z$ and satisfies $Q \leftlex P_1[x' \to y]$.
Let $q \in V(H)$ be the vertex of $H$ such that $Q[x' \to q]$ is the longest common prefix of $Q$ and $P_1[x' \to y]$.
Let $r$ be the first vertex on $Q[q \to y]$ (not including $q$) such that $r \in V(P_1)$.

Analogously as in the proof of Claim~\ref{claim:leftmost-green}, we show that the walk $P_1'$ defined in $H$ as follows:
\[ P_1' = P_1[x \to x'] \circ Q[x' \to r] \circ P_1[r \to y] \]
is a~simple path of length at least $k$ that lies lexicographically left of $P_1$ in $H$.
This leads to a contradiction to property $(e)$ of $(P_1, P_2)$.
\end{proof}

The observations above lead to an algorithm for \probb.
First, consider the case where in the solution $(P_1, P_2)$, we have that $P_2 \subseteq \rightarea{P_1}$.
After the random assignment of colors to vertices of $H$, we iterate through all vertices $u \in V(H) \setminus \{x, y, t\}$ as candidates for the vertex $x'$.
Then, we find the leftmost path $S$ from $x$ to $u$ in $H$ whose all vertices are green.
Next, we construct $\Hsplit$ and $H'$, and find the leftmost path $T$ from $u$ to $y$ in $H'$ that avoids any new vertices introduced to $H'$ by the split.
Finally, we set $P_1 = S \circ T$.
Then it only remains to verify whether there exists any path $P_2$ from $y$ to $t$ in $H$ that is internally disjoint with $P_1$.
All these checks can be done by running proper depth-first searches in total linear time.

Next, we consider the case where $P_2 \subseteq \leftarea{P_1}$.
Hence, we repeat all the above steps with any leftmost paths in $H$ replaced with the corresponding rightmost paths.
To obtain the constant probability of error, we repeat the entire process above $2^{\Oh(k)}$ times.

To derandomize this algorithm, instead of assigning the colors to vertices at random, we construct and use an~$(n, 2k)$-universal set \cite{Splitters}.
In our case this universal set is a family $\mathcal{U}$ of subsets of $V(H)$ such that for any subset $A$ of $V(H)$ of size $2k$ the family $\{ A \cap U \mid U \in \mathcal{U} \}$ contains all $2^{2k}$ subsets of $A$.
We interpret each element of $\mathcal{U}$ as a single coloring, that is, a subset $A \in \mathcal{U}$ corresponds to the coloring $c$ which assigns color green to the vertices of $A$, and blue to the rest.

We know that there is such family~$\mathcal{U}$ of size $2^{\Oh(k)} \cdot \log n$ and it can be constructed in time $2^{\Oh(k)} \cdot n \log n$ \cite{Splitters}.
We construct it once in the algorithm, and instead of drawing random colorings, we iterate through all elements of the constructed universal set~$\mathcal{U}$.
By the property of $\mathcal{U}$ there is an element $U \in \mathcal{U}$ such that $V(P_1[x \to x']) \subseteq U$ and $V(P_1[y' \to y]) \cap U = \emptyset$ because $|V(P_1[x \to x']) \cup V(P_1[y' \to y])| = 2k$. 
This guarantees the correctness of our deterministic algorithm.
\end{proof}

\paragraph*{Acknowledgement.} The authors thank Olek Łukasiewicz for pointing us to the work of Wu and Wang~\cite{WuW15}.

\bibliographystyle{plain}
\bibliography{references}
\end{document}